\documentclass[aps,twocolumn,superscriptaddress,nofootinbib,a4paper,longbibliography]{revtex4-2}
\usepackage{times}
\usepackage{epsfig}
\usepackage{amsfonts}
\usepackage{amsmath}
\usepackage{amsthm}
\usepackage{amssymb}					
\usepackage{amsthm}
\usepackage{dsfont}
\usepackage{bm}
\usepackage{mathtools}
\usepackage{color}
\usepackage{multirow}
\usepackage[normalem]{ulem}
\usepackage{latexsym}
\usepackage{mathrsfs}
\usepackage{natbib}
\usepackage{verbatim}
\usepackage[T1]{fontenc}
\usepackage{graphicx}
\usepackage{xcolor}
\usepackage{physics}
\usepackage{textcomp}
\usepackage{tikz-cd}

\usepackage{floatrow}
\usepackage{subfig}

\usepackage{caption}
 \usepackage{pifont}
\newcommand{\vmark}{\text{\ding{51}}}
\newcommand{\xmark}{\text{\ding{55}}}

\usepackage[colorlinks=true,linkcolor=blue,citecolor=magenta,urlcolor=blue]{hyperref}
\newcommand{\bracket}[3]{\langle#1|#2|#3\rangle}

\newcommand{\id}{\openone}

\newtheorem{result}{Result}
\newtheorem{conjecture}{Conjecture}
\newtheorem{definition}{Definition}
\newtheorem{lemma}{Lemma}
\newtheorem{corollary}{Corollary}

\begin{document}


\title{Correspondence between entangled states and entangled bases
    under local transformations}


\author{Florian Pimpel}\thanks{These authors contributed equally.}
\affiliation{Atominstitut,  Technische  Universit{\"a}t  Wien, Stadionallee 2, 1020  Vienna,  Austria}
\author{Martin J.~Renner}\thanks{These authors contributed equally.}
\affiliation{University of Vienna, Faculty of Physics, Vienna Center for Quantum Science and Technology (VCQ), Boltzmanngasse 5, 1090 Vienna, Austria}
\affiliation{Institute for Quantum Optics and Quantum Information - IQOQI Vienna, Austrian Academy of Sciences, Boltzmanngasse 3, 1090 Vienna, Austria}
\author{Armin Tavakoli}
\affiliation{Physics Department, Lund University, Box 118, 22100 Lund, Sweden}

\begin{abstract}
We investigate whether pure entangled states can  be associated to a measurement basis in which all vectors are local unitary transformations of the original state. We prove that for bipartite states with a local dimension that is either $2, 4$ or $8$, every state corresponds to a basis. Via numerics we strongly evidence the same conclusion also for two qutrits and three qubits. However, for some states of four qubits we are unable to find a  basis, leading us to conjecture that not all quantum states admit a corresponding measurement. Furthermore, we investigate whether there can exist a set of local unitaries that transform \textit{any} state into a basis. While we show that such a state-independent construction cannot exist for general quantum states, we prove that it does exist for real-valued $n$-qubit states if and only if $n=2,3$, and that such constructions are impossible for any multipartite system of an odd local dimension. Our results suggest a rich relationship between entangled states and iso-entangled measurements with a strong dependence on both particle numbers and dimension. 
\end{abstract}


\maketitle


\section{Introduction}
Entanglement is a fundamental, broadly useful and an intensely studied feature of quantum mechanics. However, in spite being of arguably similar foundational significance, much less is known about the entanglement of joint quantum measurements than the entanglement of quantum states. Entangled measurements are crucial for seminal quantum information protocols such as teleportation \cite{Bennett1993}, dense coding \cite{Bennett1992} and entanglement swapping \cite{Zukowski1993}, which are instrumental for various quantum technologies. Typically, they are based on the paradigmatic Bell basis, which is composed of the four maximally entangled states  $(\ket{00}\pm \ket{11})/\sqrt{2}$ and $(\ket{01}\pm\ket{10})/\sqrt{2}$. In the same way that the Bell basis may be thought of as the measurement corresponding to the maximally entangled state, it is natural to ask whether entangled states in general can be associated with a corresponding entangled measurement. Studying the relationship between entangled states and entangled measurements is not only interesting for understanding quantum mechanics. It is also an invitation to explore, in the context of quantum information applications, the largely uncharted terrain of entangled measurements beyond the Bell basis and its immediate generalisations. Most notably, entangled measurements beyond the Bell basis are also increasingly interesting for topics such as network nonlocality \cite{Tavakoli2022} and entanglement-assisted quantum communication \cite{Tavakoli2021, Pauwels2022}.

Consider that we are given a pure quantum state $\ket{\psi}$ comprised of $n$ subsystems, each of dimension $d$. Is it possible to find a measurement, namely an orthonormal basis of the global $d^n$-dimensional Hilbert space, in which all basis states have the same degree of entanglement as $\ket{\psi}$? Specifically, we want to decide the existence of $d^n$ strings, $\{V_j\}_{j=1}^{d^n}$, of local unitary transformations, 
\begin{equation}\label{unitaries}
V_j=\bigotimes_{k=1}^n U_k^{(j)}
\end{equation}
where $U_{k}^{(j)}$ is a $d$-dimensional unitary operator, such that the set of states $\ket{\psi_j}\equiv V_j\ket{\psi}$ form a basis, i.e.~$|\braket{\psi_j}{\psi_{j'}}|=\delta_{jj'}$. If affirmative, we say that $\ket{\psi}$ admits a basis and we call the set of basis vectors $\{\ket{\psi_j}\}_{j=1}^{d^n}$ a $\ket{\psi}$-basis.

Known examples of entangled measurements can be accommodated in this picture. For example, the Bell basis can be obtained from operating on $\ket{\psi}=(\ket{00}+\ket{11})/\sqrt{2}$ with the four strings of local unitaries $\{V_j\}_{j=1}^4=\{\openone\otimes \openone, \openone\otimes X, Z\otimes \openone,Z\otimes X\}$, where $X$ and $Z$ are bit-flip and phase-flip Pauli operators. A well-known generalisation of the Bell basis to $n$ systems of dimension $d$ can be thought of as a $\ket{\text{GHZ}_{n,d}}$-measurement where the relevant state is the higher-dimensional GHZ state $\ket{\text{GHZ}_{n,d}}=\frac{1}{\sqrt{d}}\sum_{k=0}^{d-1} \ket{k}^{\otimes n}$. The corresponding strings of local unitaries are $V_j=Z_d^{j_1}\otimes X_d^{j_2}\otimes \ldots \otimes X_d^{j_n}\ket{\text{GHZ}_{n,d}}$ where $j=j_1\ldots j_n \in\{0,\ldots,d-1\}^n$ and where $Z_d=\sum_{l=0}^{d-1} e^{\frac{2\pi i}{d}l}\ketbra{l}{l}$ and $X_d=\sum_{l=0}^{d-1}\ketbra{l+1}{l}$ are generalised Pauli operators. 
More generally, any state that is locally maximally entanglable  is known to admit a basis via suitable unitaries of the form $V_j=U_1^{j_1}\otimes  \ldots\otimes U_n^{j_n}$ \cite{Kruszynska2009}. These states are characterised by the property that if each qubit is supplemented with a qubit ancilla and controlled unitary gates are performed on the state-ancilla pairs, then a maximally entangled bipartite state can be constructed between the collection of state-qubits and the collection of ancilla-qubits. However, this is far from a complete characterisation of the states that admit a basis, which is seen already in the restrictive form of the strings of unitaries. For example, the three-qubit $W$-state, $\ket{W_3}=(\ket{001}+\ket{010}+\ket{100})/\sqrt{3}$, is not locally maximally entanglable but is neverthelss known to admit a basis \cite{Miyake2005}.
The subject of local unitary equivalence has also been studied by investigating whether two quantum states are LU-equivalent and by introducing a method to determine the connecting unitary \cite{KrausB2010}.  In what follows, we systematically explore a related but different question, namely whether entangled states admit an entire local-unitary equivalent orthonormal basis and, as we will introduce later, whether such bases can be constructed even without prior knowledge of the state.

\section{Corresponding Bases for the simplest quantum states}

Let us begin with considering the simplest situation, namely when $\ket{\psi}$ is a state of two qubits. We constructively show that every such state admits a basis. To this end, we first apply the state-dependent local unitaries  $W_\psi^A\otimes W^B_\psi$ that map $\ket{\psi}$, via a Schmidt decomposition, into the computational basis, $\ket{\psi_S}=\lambda \ket{00}+\sqrt{1-\lambda^2}\ket{11}$ for some coefficient $0\leq \lambda\leq 1$. Then, we consider the action of the following four strings of local unitaries 
\begin{equation}\label{2qubit}
\begin{Bmatrix}
\openone\otimes \openone\\
\openone \otimes XZ\\
XZ\otimes Z\\
XZ\otimes X
\end{Bmatrix}.
\end{equation}
One can verify that this transforms $\ket{\psi_S}$ into a $\ket{\psi}$-basis. Notice that once the state has been rotated into the Schmidt form $\ket{\psi_S}$, the subsequent unitaries \eqref{2qubit} do not depend on $\lambda$. This construction can be extended to bipartite ($n=2$) states of local dimension $d=4$ and $d=8$. Again via Schmidt decomposition, we can find state-dependent local unitaries that transform $\ket{\psi}$ into  $\ket{\psi_S}=\sum_{l=0}^{d-1} \lambda_l\ket{ll}$ for some Schmidt coefficients $\sum_l \lambda_l^2=1$. 
\begin{result}
	All bipartite states $\ket{\psi}$ of two subsystems with local dimension $d=2$, $d=4$ and $d=8$ admit a $\ket{\psi}$-basis under local unitary transformations. For states in Schmidt form $\ket{\psi_S}$, the local unitaries are independent of the specific Schmidt coefficients.
\end{result}
The proof is presented in Appendix~\ref{AppPower2}.

A seemingly very different situation emerges for two qutrits, $(n,d)=(2,3)$. In this case we fail to find strings of local unitaries that bring the Schmidt decomposition $\ket{\psi_S}$ into a basis independently of the Schmidt coefficients. Nevertheless, a basis might still be possible to construct by letting the local unitaries depend on the Schmidt coefficients. Actually, this seems to always be possible. To arrive at this, we have used a numerical method.  Let $\{\ket{\phi_j}\}_{j=1}^m$ be a set of states in a given Hilbert space. These states are pairwise orthogonal if and only if they realise the global minimum (zero) of the following objective function
\begin{equation}\label{framepot}
f(\{\phi_j\})\equiv \sum_{j \neq j'} |\braket{\phi_j}{\phi_{j'}}|^2.
\end{equation}
For a given state $\ket{\psi}$, we numerically minimise $f(\{\psi_j\})$ over all possible strings $\{V_j\}_{j=1}^{d^n}$ of local unitaries. To this end, we parameterise the local unitaries $U_k^{(j)}$ using the scheme of Ref.~\cite{Spengler2012}. For the two-qutrit case, we have randomly chosen 1000 pairs of Schmidt coefficients $(\lambda_1,\lambda_2)$ which (up to local unitaries) fully specifies the state. In each case we numerically minimise $f(\{\psi_j\})$. Without exception, we find strings of local unitaries that yield a result below our selected precision threshold of $f\leq 10^{-6}$.

Furthermore, we have also numerically investigated the case of three qubits, $(n,d)=(3,2)$. This scenario requires a different approach than the previous cases since multipartite states have no Schmidt decomposition. Instead, for any given three-qubit state $\ket{\psi}$, there exists local unitary transformations that map it onto the canonical form $a\ket{000}+b\ket{011}+c\ket{101}+d\ket{110}+e\ket{111}$ where $(b,c,d,e)$ are real numbers and $a$ is a complex number \cite{Acin2000, Caretet2000}. Hence, up to local unitaries, the state space (after normalisation) is characterised by five real numbers. Later, we will provide an analytical construction of a $\ket{\psi}$-basis for the four-parameter family corresponding to restricting $a$ to be real. However, we have not found an analytical basis construction for general three-qubit states, but we nevertheless conjecture that it exists. To evidence this, we have employed the previously introduced numerical search method. Again, we have randomly chosen 1000 normalised sets of coefficients $(a,b,c,d,e)$ and searched for the minimal value of $f$ over all the strings of local qubit unitaries. In all cases, we find that $f$ vanishes up to our selected precision of $f\leq 10^{-6}$.
In summary we are left with
\begin{conjecture}
	All states $\ket{\psi}$ of two qutrits ($(n,d)=(2,3)$) and three qubits ($(n,d)=(3,2)$)  admit a $\ket{\psi}$-basis under local unitary transformations.
\end{conjecture}

\section{Existence of states without a corresponding basis}
Given the above case studies, one might suspect that every pure quantum state admits a basis. Interestingly, this seems not to be true. While some states of four qubits, $(n,d)=(4,2)$, are found to admit a basis, for example a Dicke state \cite{Tanaka2007}, it appears that most four-qubit states do not admit a basis.
\begin{conjecture}
\label{conj:nonexist}
	There exist four-qubit states that do not admit a basis under local unitary transformations. 
\end{conjecture}
We have sampled many different four-qubit states and repeatingly attempted to numerically find a basis via the minimisation of \eqref{framepot}, also using several different search algorithms. It was regularly found that the estimated minimum is multiple orders of magnitude above our given precision threshold for a basis. For example, we searched for the minimum of $f$ for the state $\frac{2}{\sqrt{6}}\ket{W}+\frac{\sqrt{2}}{\sqrt{6}}\ket{\text{GHZ}_{4,2}}$, with 100 randomised initial points, and never reached below $f=10^{-1}$, five orders of magnitude above our precision threshold. We have attempted to prove that no basis exists by employing semidefinite outer relaxations of $f$ over the set of dimensionally-restricted quantum correlations \cite{Navascues2015} combined with a modified sampling of the state and measurement space \cite{Distrust2021} and symmetrisation techniques \cite{Tavakoli2019Sym} to efficiently treat the large number  of single-qubit unitaries featured in this problem. However, the conjecture has resisted our efforts. A guiding intuition is that the number of free parameters is $3n(2^n-1)$ whereas the number of orthogonality constraints (both real and imaginary part) is $2^{2n}-2^{n}$, and the latter is larger than the former only when $n\geq 4$. Also, we numerically minimized $f$ for the same state, with a lesser amount of required orthonormal states than a full basis. The numerics suggest that $N=12$ orthonormal states are possible to find, while for $N=13$ we never reached below $f=10^{-3}$. This is also in line with the above parameter-counting-argument. For $N$ orthonormal states and $n=4$ qubits, we have $12(N-1)$ free parameters and $N^2-N$ orthogonality constraints. The latter is larger than the former for $N \geq 13$.
This also has consequences for the existence of states that admit a basis for an arbitrary number of qubits.

\begin{result}
	If Conjecture~\ref{conj:nonexist} is true, namely that some four-qubit states do not admit a basis, then the same holds for states of any number of qubits bigger than four.
\end{result}

\begin{proof}
We show, that if an $n$-qubit state $\ket{\psi}$ does not admit a basis, then the $(n+1)$-qubit state $\ket{\psi'}=\ket{\psi}\otimes \ket{0}$ also does not admit a basis. This inductively validates the above theorem. By contradiction, suppose there are $2^{n+1}$ unitaries $V'_j=V_j\otimes U^{(j)}_{n+1}$ such that $|\bracket{\psi'}{(V'_j)^\dagger V'_{k}}{\psi'}|=\delta_{jk}$ $\forall j,k\in \{1,...,2^{n+1}\}$. Divide the $2^{n+1}$ states $U^{(j)}_{n+1}\ket{0}$ into two sets such that two orthogonal vectors are not in the same set (e.g.~the northern and southern hemisphere of the Bloch ball). Consider the set that contains at least as many elements as the other one, hence, at least $2^{n}$ elements. By construction, these states cannot be distinguished on the last qubit, $|\bracket{0}{U_{n+1}^{(j)\dagger} U^{(k)}_{n+1}}{0}|\neq 0$. Since $|\bracket{\psi'}{(V'_j)^\dagger V'_{k}}{\psi'}|=|\bracket{\psi}{V_j^{\dagger} V_{k}}{\psi}|\cdot |\bracket{0}{U_{n+1}^{(j)\dagger} U^{(k)}_{n+1}}{0}|$, we must have $|\bracket{\psi}{V_j^{\dagger} V_{k}}{\psi}|=\delta_{jk}$ for all of those pairs, which contradicts that $\ket{\psi}$ does not admit a basis. 
\end{proof}

\section{Special families of states with a basis}
Since not all pure quantum states admit a basis, and this seems to be typical rather than exceptional for four qubits, it is interesting to ask whether some distinguished families of $n$-qubit states can nevertheless admit a basis. This is well-known to be the case for $n$-qubit GHZ-states and graph-states since they are locally maximally entanglable. More interestingly, a positive answer is also possible for states that are not of this kind: we construct a basis for the $n$-qubit $W$-state, $\ket{W_n}=\frac{1}{\sqrt{n}}\sum_{\sigma} \sigma(\ket{0}^{\otimes n-1} \ket{1})$ where $\sigma$ runs over all permutations of the position of ``$1$''. Note that $\ket{W_1}=\ket{1}$ and that a $\ket{W_1}$-basis is obtained from the unitaries $\{\openone, X\}$. Now we apply induction. Consider that the strings $\{V_j^{(n)}\}_{j=1}^{2^n}$ generate a $\ket{W_n}$-basis. One can then construct a basis for $n+1$ qubits as follows. For half of the basis elements, namely $j=1,\ldots,2^n$, define $V_j^{(n+1)}=V_j^{(n)} \otimes \openone$ and for the other half, namely $j=2^n+1,\ldots, 2^{n+1}$, define $V_j^{(n+1)}=\bigotimes_{k=1}^n U_k^{(j)}Z \otimes X$. As we detail in Appendix~\ref{AppW}, one can verify that $\{V_j^{(n+1)}\ket{W_{n+1}}\}_j$ is a $W$-basis. We note that for the purpose of entanglement distillation, a different construction of a $W$-basis was given in Ref.~\cite{Miyake2005}.

\section{State-independent basis construction}
So far, we have considered whether a specific state can be associated to a specific measurement. In other words, the unitary constructions have been state-dependent. We now go further and introduce a complementary concept, namely whether state-independent basis construction exist.

\begin{definition}
	Strings of local unitary transformations $\{V_j\}_{j=0}^{d^n}$ that can transform any state in a space of states $\mathcal{S}$ into a basis, i.e.~strings of local unitaries that satisfy 
	\begin{align*}
&\forall \psi\in\mathcal{S}, \hspace{-5mm} &&|\bracket{\psi}{V_j^\dagger V_{j'}}{\psi}|=\delta_{jj'}
\end{align*}
are called state-independent basis constructions.
\end{definition}

Naturally, this state-independent notion of basis construction is much stronger than the previously considered state-dependent notion. In the most ambitious case, we choose the space $\mathcal{S}$ to be the entire Hilbert space of $n$ subsystems of dimension $d$, i.e.~$\mathcal{S}\simeq (\mathbb{C}^d)^{\otimes n}$.
\begin{result}
	If $\mathcal{S}\simeq (\mathbb{C}^d)^{\otimes n}$ then a state-independent basis construction cannot exist.
\end{result}
\begin{proof} 
In fact, not even two orthogonal vectors can be state-independently constructed for the full quantum state space. To show this, we can w.~l.~g.~set $V_1=\openone$ and assume that there exists local unitaries $\{U_k\}$ such that $\ket{\psi_1}=\ket{\psi}$ and $\ket{\psi_2}=\bigotimes_{k=1}^n U_k\ket{\psi}$ are orthogonal for all $\ket{\psi}$. Focus now on the particular state $\ket{\psi}=\bigotimes_{k=1}^n \ket{\mu_k}$ where $\ket{\mu_k}$ is some eigenvector of the unitary $U_k$. Since the eigenvalues of a unitary are complex phases, written $e^{i\varphi_k}$ for $U_k$ and $\ket{\mu_k}$, we obtain $\ket{\psi_1}=\bigotimes_{k=1}^n \ket{\mu_k}$ and $\ket{\psi_2}=e^{i\sum_{k=1}^n \varphi_k}\bigotimes_{k=1}^n \ket{\mu_k}$. These two states are evidently not orthogonal and hence we have a contradiction.
\end{proof}

Interestingly, the situation changes radically if we limit our state-independent investigation to all quantum states in a real-valued Hilbert space. That is,  $\mathcal{S}\simeq (\mathbb{R}^d)^{\otimes n}$.  Such real quantum systems have also been contrasted in the literature with their complex counterparts \cite{Renou2021, Wootters2012, Wu2021}.  Let us momentarily ignore the $n$-partition structure of our Hilbert space and simply consider two real states, connected by a unitary.
\begin{lemma}
	A unitary transformation $U$ maps every real state  $\ket{\psi}$ to an orthogonal state  $U\ket{\psi}$, i.e. $\bracket{\psi}{U}{\psi}=0$, if and only if  $U$ is skew-symmetric ($U=-U^T$). \label{lemma1}
\end{lemma}
\begin{proof}

First, we assume the skew-symmetry property $U=-U^T$. For real states $\ket{\psi}$ we obtain $\bracket{\psi}{U}{\psi}=\bracket{\psi}{U^\dagger}{\psi}^*=\bracket{\psi}{U^T}{\psi}$. Using the skew-symmetry property $U=-U^T$, this implies $\bracket{\psi}{U}{\psi}=\bracket{\psi}{U^T}{\psi}=-\bracket{\psi}{U}{\psi}$, hence $\bracket{\psi}{U}{\psi}=0$. Conversely, assume that $\bracket{\psi}{U}{\psi}=0$ for all real-valued $\ket{\psi}$. Choosing in particular $\ket{\psi}=\ket{k}$ for $k=0,\ldots,d-1$, it follows that all diagonal elements of $U$ must vanish. Then, choose $\ket{\psi}=\frac{1}{\sqrt{2}}(\ket{i}+\ket{j})$ for any pair $i\neq j$. This yield $U_{ii}+U_{jj}+U_{ij}+U_{ji}=0$, but since we know that the diagonals vanish we are left with just $U_{ij}=-U_{ji}$ which defines a skew-symmetric operator.
\end{proof}

Returning to our $n$-partitioned real Hilbert space, and still w.~l.~g.~taking $V_1=\openone$, the above result demands that we find local unitaries such that 
\begin{equation}
U_1\otimes \ldots \otimes U_n=- U_1^T\otimes \ldots \otimes U_n^T.
\end{equation}
This is only possible if $U_k^T=\pm U_k$. Hence, all local unitaries must be either symmetric or skew-symmetric, and the number of the latter must be odd. When extended from two orthogonal states to a whole basis, we require that this property holds for every pair of distinct labels $(j,j')$ in the basis. 
\begin{corollary}
\label{cor:skewsym}
 The $d^n$ strings of local unitaries $\{V_j\}_{j=1}^{d^n}$ form a state-independent basis construction on $\mathcal{S}\simeq (\mathbb{R}^d)^{\otimes n}$ if and only if $(V_j)^\dagger V_{j'}$ is skew-symmetric for every $j\neq j'$.
\end{corollary}
\begin{proof}
    Since the strings of local unitaries $\{V_j\}$ shall form a state-independent basis construction, the states $\ket{\psi_j}=V_j \ket{\psi}$ have to be pairwise orthogonal for every real state $\ket{\psi}$. Therefore $\bra{\psi}(V_j)^\dagger V_{j'}\ket{\psi}=0$ for all $j\neq j'$ and for all real states $\ket{\psi}$. By Lemma~\ref{lemma1}, all of these $(V_j)^\dagger V_{j'}$ have to be skew-symmetric. At the same time, if all $(V_j)^\dagger V_{j'}$ are skew-symmetric, the states $\ket{\psi_j}$ are pairwise orthogonal and form a basis.
\end{proof}

The question becomes whether the above condition can be satisfied for a given scenario. Consider it first for qubit systems ($d=2$).
\begin{lemma}
\label{lemma:paulireduction}
	In qubit-systems the set of complex local unitaries that are either symmetric or skew-symmetric and whose products are again either symmetric or skew-symmetric, must obey a simple structure; they are equivalent to the four Pauli-type operators $\mathcal{P}\equiv \{\openone, X,Z,XZ\}$.
\end{lemma}
 The proof is presented n Appendix~\ref{AppPauliStructure}. Thus, if a state-independent construction exists, we can restrict to selecting one of these four operators for each of our local unitaries $U_k^{(j)}$. Interestingly, 
\begin{result}
	In the cases of real two-, and three-qubits, a state-independent construction is possible.
\end{result}
\begin{proof}
For two qubits $(n,d)=(3,2)$, it is in fact given by Eq.~\eqref{2qubit}. One can straightforwardly verify that the criterion in Corollary~\ref{cor:skewsym} is satisfied, i.e.~all local unitaries are selected from $\mathcal{P}$ and all pairs of products of unitary strings in \eqref{2qubit} are skew-symmetric. Alternatively, one can easily verify that \eqref{2qubit} maps every state $\sum_{i,j=0,1}\alpha_{ij}\ket{ij}$ into a basis, for any real coefficients $\alpha_{ij}$.
Furthermore, by the same token, an explicit state-independent basis construction for every real state of three qubits, $(n,d)=(3,2)$, that satisfies our necessary and sufficient criterions, is given by the following set of eight strings of local unitaries
\begin{equation}\nonumber
\begin{Bmatrix}
&\openone \otimes \openone\otimes \openone\\\nonumber
& Z\otimes Z \otimes XZ\\\nonumber
& Z \otimes XZ\otimes \openone\\\nonumber
& XZ \otimes \openone\otimes \openone\\\nonumber
& Z \otimes X\otimes XZ\\\nonumber
& X \otimes \openone\otimes XZ\\
& X \otimes XZ\otimes Z\\
& X \otimes XZ\otimes X
\end{Bmatrix}.
\end{equation}
Again, one may easily verify that every real state $\sum_{i,j,k=0,1}\alpha_{ijk}\ket{ijk}$ is mapped into a basis.
\end{proof}

\begin{table*}[ht]
	\centering
	\begin{tabular}{|c|c|c|c|c|c|c|c|c|}
		\hline
		& (2,2,$\mathbb{R}$) & (2,2,$\mathbb{C}$) & (3,2,$\mathbb{R}$) & (3,2,$\mathbb{C}$) & (4,2,$\mathbb{R}$) &  (2,3,$\mathbb{C}$) & (2,$4$ or $8$,$\mathbb{C}$) & ($n$, $2m+1$,$\mathbb{R}$) \\ \hline
		\begin{tabular}[c]{@{}c@{}}State-dependent\\ construction\end{tabular}   &   \vmark                              &  \vmark                                & \vmark        & (\vmark)        &  (\xmark)       & (\vmark)            &    \vmark            & $---$            \\ \hline
		\begin{tabular}[c]{@{}c@{}}State-independent\\ construction\end{tabular} &    \vmark                             &  \xmark                               &    \vmark     & \xmark        &  \xmark       &  \xmark               &           \xmark      & \xmark           \\ \hline
	\end{tabular}
	\caption{Overview of results. The first row indicates the scenario: $(n,d,\mathcal{S})$  gives particle number, dimension and the type of state space respectively. The symbol \vmark indicates the existence of a basis under local unitaries. The symbol \xmark indicates that there in general can be no basis under local unitaries, i.e.~at least one state admits no basis. Paranthesis indicates that the result is obtained from numerical search. The symbol $---$ indicates that no investigation was made.}\label{tabresults}
\end{table*}

Two- and three-qubits are interesting cases because they are exceptional. 
\begin{result}
There exist no state-independent construction for real states of four or more qubits. 
\end{result}
\begin{proof}
We first prove this for $n=4$ and then show that this implies impossibility also for $n>4$. The four-qubit case contains 16 strings of unitaries and by Lemma~\ref{lemma:paulireduction} we know that each local unitary can w.~l.~g.~be selected from $\mathcal{P}$. Since we seek a state-independent construction, we can momentarily consider only the state $\ket{0000}$. In order for it to be mapped into a basis, we see that $Z$ acts trivially on every register and therefore  each one of the 16 combinations of bit-flip or identity operators, $\{X^{c_1}\otimes X^{c_2} \otimes X^{c_3} \otimes X^{c_4}\}$ for $c_1,c_2,c_3,c_4\in\{0,1\} $, must be featured in exactly one of the 16 unitary strings $\{V_j\}_{j=1}^{16}$. Let us now look only at six of these strings, namely those corresponding to having zero bit-flips (1 case), one bit-flip (4 cases) and four bit-flips (1 case). W.~l.~g.~fixing $V_1=\openone$ (zero bit-flips), the strings take the form
\begin{equation}
\begin{array}{c||cccccccc}
V_1 & \id & \otimes & \id & \otimes & \id & \otimes & \id\\
V_2 & XZ^{r_{11}} & \otimes & Z^{r_{12}} & \otimes & Z^{r_{13}} & \otimes & Z^{r_{14}}\\
V_3 & Z^{r_{21}} & \otimes & XZ^{r_{22}} & \otimes & Z^{r_{23}} & \otimes & Z^{r_{24}}\\
V_4 & Z^{r_{31}} & \otimes & Z^{r_{32}} & \otimes & XZ^{r_{33}} & \otimes & Z^{r_{34}}\\
V_5 & Z^{r_{41}} & \otimes & Z^{r_{42}} & \otimes & Z^{r_{43}} & \otimes & XZ^{r_{44}}\\
V_6 & XZ^{r_{51}} & \otimes & XZ^{r_{52}} & \otimes & XZ^{r_{53}} & \otimes & XZ^{r_{54}}
\end{array},
\label{4qubit}
\end{equation}
where $r_{ij}\in\{0,1\}$ represent our freedom to insert a $Z$ operator and thus realise the two relevant elements of $\mathcal{P}$. Since every row must be skew-symmetric and the only skew-symmetric element in $\mathcal{P}$ is $XZ$, we must have $r_{11}=r_{22}=r_{33}=r_{44}=1$ and $r_{51}+r_{52}+r_{53}+r_{54}=1$ where addition is modulo two. Moreover, every product of two rows must be skew-symmetric, i.e.~the product must have an odd number of $XZ$ operations. For the four middle rows, this implies $r_{ij}+r_{ji}=1$ for distinct indices $i,j\in\{1,2,3,4\}$. For the products $V^\dagger_6V_j$ for $j=2,3,4,5$, the conditions for skew-symmetry respectively become
\begin{align}\nonumber
& r_{12}+r_{13}+r_{14}+r_{52}+r_{53}+r_{54}=1\\\nonumber
& r_{21}+r_{23}+r_{24}+r_{51}+r_{53}+r_{54}=1\\\nonumber
& r_{31}+r_{32}+r_{34}+r_{51}+r_{52}+r_{54}=1\\
& r_{41}+r_{42}+r_{43}+r_{51}+r_{52}+r_{53}=1.
\end{align}
Summing these four equations and using the previously established skew-symmetry conditions, one can cancel out  all degrees of freedom $r_{ij}$ and arrive at the contradiction $1=0$. Hence, we conclude that the state-independent basis construction for four qubits is impossible. 

For the case of five qubits, we can again assume w. l. g. that the 32 combinations of bit-flip or identity operators, $\{X^{c_1}\otimes X^{c_2} \otimes X^{c_3} \otimes X^{c_4} \otimes X^{c_5}\}$ for $c_1,c_2,c_3,c_4,c_5\in\{0,1\} $ must be featured in exactly one of the 32 unitary strings since the state $\ket{00000}$ has to be mapped into an orthonormal basis. Suppose there is a state-independent construction that maps every real-valued five-qubit state into a basis, in especially any state of the form $\ket{\psi}\otimes \ket{0}$, where $\ket{\psi}$ is an arbitrary real-valued four qubit state. Now consider the 16 strings where $c_5=0$. Since the fifth qubit is always mapped to itself, it has to hold that the first four qubits are pairwise distinguishable. However, this implies a state-independent construction for four qubits which is in contradiction to the above. By induction, this implies that no state-independent construction can exist whenever $n\geq~4$.
\end{proof}

The possibility of state-independent constructions for real-valued bi- and tri-partite systems draws heavily on the simple structure of skew-symmetric qubit unitaries. If we consider real-valued systems of dimension $d>2$, the situation changes considerably. 
\begin{result}
	State-independent constructions are impossible for all states of systems with odd local dimensions, i.e.~when $(n,d)=(n,2m+1)$.
\end{result}
\begin{proof}
This stems from the fact that there exists no skew-symmetric unitary matrix in odd dimensions. To see that, simply note that if $A$ is skew-symmetric then $\det(A) = \det(A^T) = \det(-A) =  (-1)^{2m+1} \det(A)=-\det(A)$ and hence $\det(A)=0$, but that contradicts unitarity because the determinant of a unitary has unit modulus.  
\end{proof}

\section{Conclusion}

In summary, we have investigated the correspondence between entangled states and entangled measurements under local unitary transformations, both when the local transformation can and cannot explicitly depend on the target state. Perhaps surprisingly, we have found that this problem is not so straightforward and has a strong dependence on both the number of subsystems involved and their dimension.  Our analytical and numerical results and conjectures are summarised in Table~\ref{tabresults}.

The main open problem left by our work is to prove that there exists four-qubit states that admit no basis under local unitaries. A related question is to bound the volume of such states. Moreover, for the state-independent considerations, we focused on real Hilbert spaces while other natural spaces are left to explore, e.g.~states with a fixed entanglement entropy or symmetric $n$-qubit subspaces. Notably, questions of this type naturally continue recent efforts to explore the role of entangled measurements (beyond the Bell basis) in quantum correlation scenarios \cite{Gisin2019, Branciard2021, Tang2020, Huang2022, Bäumer2021}. For example, a notable shortcomming of the traditional multiqubit entanglement swapping protocol is that the loss of one particle renders the measurement separable. However, other states can preserve a degree of entanglement under reductions. The existence of an iso-entangled basis of such states is an invitation to investigate noise-resiliant entanglement swapping protocols which are important building blocks for many quantum information applications.

\textit{Note added.---} During the late stage of our work, we became aware of the previous work \cite{Tanaka2007} where i.~a.~bases are found for some Dicke states. 

\begin{acknowledgments}
We thank Hayata Yamasaki, Marcus Huber, Jakub Czartowski and Karol \.{Z}yczkowski for discussions. A. T.~acknowledges support from the Wenner-Gren Foundation and from the Knut and Alice Wallenberg Foundation through the Wallenberg Center for Quantum Technology (WACQT). M. J. R. acknowledges financial support from the Austrian Science Fund (FWF) through BeyondC (F7103-N38), the Project No. I-2906, as well as support by the John Templeton Foundation through Grant 61466, The Quantum Information Structure of Spacetime (qiss.fr), the Foundational Questions Institute (FQXi) and the research platform TURIS. The opinions expressed in this publication are those of the authors and do not necessarily reflect the views of the John Templeton Foundation.
\end{acknowledgments}

\bibliography{reference_entbasis}

\appendix
\onecolumngrid

\section{Basis construction for every bipartite state of local dimension $d=4$ and $d=8$}\label{AppPower2}
Let the local dimension be a power of two, $d=2^m$, and index the $d^2$ basis elements as  $(\tilde{j},j)$ where $\tilde{j}=0,1,\ldots, d-1$ and $j=1,2, \ldots,d$. Let $W^A_\psi\otimes W^B_\psi$ be the state-dependent local unitaries that transform the general state $\ket{\psi}$ into the Schmidt basis, i.e.~$\ket{\psi_S}\equiv W^A_\psi\otimes W^B_\psi\ket{\psi}=\sum_{l=0}^{d-1} \lambda_l \ket{l,l}$, with the Schmidt coefficients $\lambda_l \in \mathbb{R}$ satisfying $\sum_l \lambda_l^2=1$. We now further decompose the individual $d$-dimensional registers as a string of $m$ qubits, writing $\ket{l}=\ket{l_1\ldots l_m}$. Thus, the Schmidt decomposed state reads
\begin{equation}\label{appeq1}
	\ket{\psi_S}=\sum_{l_1, \ldots, l_m=0,1} \lambda_l \ket{l_1\ldots  l_m,l_1\ldots l_m}.
\end{equation}
Once the state has been put in the form \eqref{appeq1}, we apply a set of local unitaries that is independent of the Schmidt coefficients. For $d=4$ and $\tilde{j}=0$, the two sets of unitaries read as follows:
\begin{equation}
	\begin{array}{cc||cc||c}
		\tilde{j}& j & U_1^{(\tilde{j},j)} & U_2^{(\tilde{j},j)} & U_1^{(\tilde{j},j)} \otimes U_2^{(\tilde{j},j)}\ket{\psi_S}\\ \hline\hline
		0 & 1 & \openone \otimes \openone & \openone\otimes \openone &\lambda_{00}\ket{00,00}+\lambda_{01}\ket{01,01}+\lambda_{10}\ket{10,10}+\lambda_{11}\ket{11,11} \\
		0 & 2 & \openone \otimes X & \openone \otimes XZ &\lambda_{00}\ket{01,01}-\lambda_{01}\ket{00,00}+\lambda_{10}\ket{11,11}-\lambda_{11}\ket{10,10} \\
		0 & 3 & X \otimes \openone & XZ\otimes Z &\lambda_{00}\ket{10,10}-\lambda_{01}\ket{11,11}-\lambda_{10}\ket{00,00}+\lambda_{11}\ket{01,01} \\
		0 & 4 & X \otimes X & XZ\otimes X &\lambda_{00}\ket{11,11}+\lambda_{01}\ket{10,10}-\lambda_{10}\ket{01,01}-\lambda_{11}\ket{00,00}
	\end{array}\label{ququart}
\end{equation}

In addition, we define $U_1^{(\tilde{j},j)}:= X_4^{\tilde{j}}\ U_1^{(\tilde{j}=0,j)}$ and $U_2^{(\tilde{j},j)}:=U_2^{(\tilde{j}=0,j)}$, where $X_d$ is the $d$-dimensional shift-operator $X_d=\sum_{l=0}^{d-1}\ketbra{l+1}{l}$. Note that, the unitaries $U_2^{(\tilde{j},j)}$ coincide with the state-independent set for two qubits given in Eq.~\eqref{2qubit} and do not depend on $\tilde{j}$. 
At the same time, $U_1^{(\tilde{j}=0,j)}$ are the same as $U_2^{(\tilde{j},j)}$ where the $Z$ gates are left out. We now show that $\{U_1^{(\tilde{j},j)} \otimes U_2^{(\tilde{j},j)}\ket{\psi_S}\}_{\tilde{j},j}$ is a basis of the bipartite Hilbert space. One can check directly that the four states with $\tilde{j}=0$ stated in Eq.~\eqref{ququart} above are pairwise orthogonal. We want to mention that we are exploiting the fact that $U_2^{(\tilde{j}=0,j)}$ are the elements of a state-independent construction. To see the connection, note that the calculation for the state-independent two-qubit construction reads as follows:
\begin{align}
	(\openone \otimes \openone) (\lambda_{00}\ket{00}+\lambda_{01}\ket{01}+\lambda_{10}\ket{10}+\lambda_{11}\ket{11})&=\lambda_{00}\ket{00}+\lambda_{01}\ket{01}+\lambda_{10}\ket{10}+\lambda_{11}\ket{11}\, ,\\
	(\openone \otimes XZ) (\lambda_{00}\ket{00}+\lambda_{01}\ket{01}+\lambda_{10}\ket{10}+\lambda_{11}\ket{11})&=\lambda_{00}\ket{01}-\lambda_{01}\ket{00}+\lambda_{10}\ket{11}-\lambda_{11}\ket{10}\, ,\\
	(XZ \otimes Z) (\lambda_{00}\ket{00}+\lambda_{01}\ket{01}+\lambda_{10}\ket{10}+\lambda_{11}\ket{11})&=\lambda_{00}\ket{10}-\lambda_{01}\ket{11}-\lambda_{10}\ket{00}+\lambda_{11}\ket{01}\, ,\\
	(XZ \otimes X) (\lambda_{00}\ket{00}+\lambda_{01}\ket{01}+\lambda_{10}\ket{10}+\lambda_{11}\ket{11})&=\lambda_{00}\ket{11}+\lambda_{01}\ket{10}-\lambda_{10}\ket{01}-\lambda_{11}\ket{00}\, .
\end{align}
Since these states are pairwise orthogonal for arbitrary real coefficients $\lambda_{l_1l_2}$, the same holds true for the states in Eq.~\eqref{ququart}.
In addition, all of the states where $\tilde{j}=0$ are elements of the subspace spanned by $\ket{00,00}$, $\ket{01,01}$, $\ket{10,10}$ and $\ket{11,11}$. Hence, they form a basis of this four-dimensional subspace. By shifting now the first system we obtain a basis for the remaining orthogonal subspaces. More precisely, since we defined $U_1^{(\tilde{j},j)}= X_4^{\tilde{j}}\ U_1^{(\tilde{j}=0,j)}$ the states where $\tilde{j}=1$ are esentially the same states as the ones in Eq.~\eqref{ququart} but with the first system shifted by one $l\rightarrow l\oplus 1\  (\text{mod } 4)$. For example, $\lambda_{00}\ket{11,10}-\lambda_{01}\ket{00,11}-\lambda_{10}\ket{01,00}+\lambda_{11}\ket{10,01}$ is the state that corresponds to $\tilde{j}=1$ and $j=3$. In this way, the four states where $\tilde{j}=1$ form a basis of the subspace spanned by $\ket{01,00}$, $\ket{10,01}$, $\ket{11,10}$ and $\ket{00,11}$ (or all states where $\ket{l+1,l}$). Analogously, the four states where $\tilde{j}=2$ ($\tilde{j}=3$) form a basis of the subspaces spanned by the vectors with $\ket{l+2,l}$ ($\ket{l+3,l}$).  Altogether, the sixteen states $\{U_1^{(\tilde{j},j)} \otimes U_2^{(\tilde{j},j)}\ket{\psi_S}\}_{\tilde{j},j}$ form a basis of the entire sixteen dimensional Hilbert space.\\

A similar construction can be found for $d=8$ by using the state-independent construction of three qubits. Similar as above, the set for $\tilde{j}=0$ reads as follows:
\begin{equation}
	\begin{array}{cc||cc||c}
		\tilde{j}& j & U_1^{(\tilde{j},j)} & U_2^{(\tilde{j},j)} & U_1^{(\tilde{j},j)} \otimes U_2^{(\tilde{j},j)}\ket{\psi_S}\\ \hline\hline
		0 & 1 & \openone \otimes \openone\otimes \openone & \openone \otimes \openone\otimes \openone &+\lambda_{000}\ket{000,000}+\lambda_{001}\ket{001,001}+\lambda_{010}\ket{010,010}+\lambda_{011}\ket{011,011}\\
		&&&& +\lambda_{100}\ket{100,100}+\lambda_{101}\ket{101,101}+\lambda_{110}\ket{110,110}+\lambda_{111}\ket{111,111} \\\hline
		0 & 2 & \openone \otimes \openone \otimes X & Z\otimes Z \otimes XZ &+\lambda_{000}\ket{001,001}-\lambda_{001}\ket{000,000}-\lambda_{010}\ket{011,011}+\lambda_{011}\ket{010,010}\\
		&&&& -\lambda_{100}\ket{101,101}+\lambda_{101}\ket{100,100}+\lambda_{110}\ket{111,111}-\lambda_{111}\ket{110,110} \\\hline
		0 & 3 & \openone \otimes X\otimes \openone & Z \otimes XZ\otimes \openone &+\lambda_{000}\ket{010,010}+\lambda_{001}\ket{011,011}-\lambda_{010}\ket{000,000}-\lambda_{011}\ket{001,001}\\
		&&&& -\lambda_{100}\ket{110,110}-\lambda_{101}\ket{111,111}+\lambda_{110}\ket{100,100}+\lambda_{111}\ket{101,101} \\\hline
		0 & 4 & X \otimes \openone\otimes \openone & XZ \otimes \openone\otimes \openone & (...) \\\hline
		0 & 5 & \openone \otimes X\otimes X & Z \otimes X\otimes XZ & (...) \\\hline
		0 & 6 & X \otimes \openone\otimes X & X \otimes \openone\otimes XZ & (...) \\\hline
		0 & 7 & X \otimes X\otimes \openone & X \otimes XZ\otimes Z  & (...) \\\hline
		0 & 8 & X \otimes X \otimes X & X \otimes XZ\otimes X & (...) 
	\end{array}
\end{equation}
Again, we define $U_1^{(\tilde{j},j)}= X_8^{\tilde{j}}\ U_1^{(\tilde{j}=0,j)}$ and $U_2^{(\tilde{j},j)}=U_2^{(\tilde{j}=0,j)}$. The proof that this forms a basis of the 64-dimension Hilbert space is completely analogous to the case of $d=4$ before. The eight states for $\tilde{j}=0$ form a basis of the eight-dimensional subspace spanned by $\ket{l_1l_2l_3,l_1l_2l_3}$ (for $l_i=0,1$). Applying the shift operator $X_8$ to the first system, one obtains bases of the other eight-dimensional orthogonal subspaces spanned by the vectors with $\ket{l+\tilde{j},l}$. This approach cannot (immediately) be generalized to higher dimensions $d=2^n$, due to the lack of state-independent constructions for $n\geq 4$ qubits. However, there is in principle no reason to restrict the unitaries on the second system to tensor products of single qubit Pauli gates as we do here. In principle, we could also consider general permutations with suitably chosen signs such that all terms cancel in this pairwise sense as above. Even when considering this larger class of possibilities, we made an exhaustive search and could not find any additional construction. Due to this, it seems unlikely that a construction exists in which the unitaries do not depend on the Schmidt coefficients.

\section{An $n$-qubit basis of $W$-states}\label{AppW}
We define the $n$-qubit $W$-state as
\begin{align}
	& \ket{W_1}\equiv  \ket{1} \nonumber \\
	&\ket{W_2}\equiv  \frac{1}{\sqrt{2}}\left(\ket{01}+\ket{10}\right) \nonumber \\
	&\ket{W_3}\equiv  \frac{1}{\sqrt{3}}\left(\ket{001}+\ket{010}+\ket{100}\right) \nonumber \\
	&\ket{W_4}\equiv \frac{1}{2}\left(\ket{0001}+\ket{0010}+\ket{0100}+\ket{1000}\right) \nonumber \\
	& \quad \vdots
\end{align}
Note that for one and two qubits, the definition is only introduced for sake of convenience. In general, we write
\begin{equation}
	\ket{W_n}\equiv\frac{1}{\sqrt{n}}\sum_{\sigma} \sigma(\ket{0}^{\otimes n-1} \ket{1}),
\end{equation}
where $\sigma$ runs over all permutations of the position of ``$1$''. It is also useful to write the state recursively as
\begin{equation}\label{Wrecursion}
	\ket{W_{n+1}}=\sqrt{\frac{n}{n+1}}\ket{W_n}\otimes \ket{0} +\frac{1}{\sqrt{n+1}}\ket{0}^{n}\otimes \ket{1}
\end{equation}

Clearly, if we apply the local unitaries $U_1^{(1)}=\openone$ and $U_1^{(2)}=X$ to $\ket{W_1}$ we generate the trivial one-qubit $W$-basis $\{\ket{0},\ket{1}\}$. Assume now that the local unitaries $\{U_k^{(j)}\}$ for $k=1,\ldots n$ and $j=1,\ldots,2^n$ yield a $\ket{W_n}$-basis. We will now show that under this assumption we can construct a basis for $\ket{W_{n+1}}$ and hence it follows from induction that a $W$-basis exists for any number of qubits.

We illustrate the induction step as follows,
\begin{equation}
	\begin{array}{@{}ccccccc|cc}
		U^{(1)}_{1} & \otimes &  U^{(1)}_{2} & \otimes & \hdots & \otimes & U^{(1)}_{n}  & \otimes & \id \\
		U^{(2)}_{1} & \otimes &  U^{(2)}_{2} & \otimes & \hdots & \otimes & U^{(2)}_{n}  & \otimes & \id \\
		&  \vdots & & & &  \vdots & & \vdots & \\
		U^{(2^n)}_{1} & \otimes &  U^{(2^n)}_{2} & \otimes & \hdots & \otimes & U^{(2^n)}_{n}  & \otimes & \id \\ \hline 
		U^{(1)}_{1} Z & \otimes &  U^{(1)}_{2} Z & \otimes & \hdots & \otimes & U^{(1)}_{n} Z  & \otimes & X \\
		U^{(2)}_{1} Z & \otimes &  U^{(2)}_{2} Z & \otimes & \hdots & \otimes & U^{(2)}_{n} Z  & \otimes &  X \\
		&  \vdots & & & &  \vdots &  & \vdots & \\
		U^{(2^{n})}_{1} Z & \otimes &  U^{(2^{n})}_{2} Z & \otimes & \hdots & \otimes & U^{(2^{n})}_{n} Z  & \otimes &  X \\
	\end{array}.
	\label{construction}
\end{equation}
We see that for the first $2^n$ basis elements, we extend the unitaries for $n$ qubits by tensoring with $\openone$ for qubit number $n+1$. For the latter $2^n$ basis elements, we extend the unitaries for $n$ qubits by multiplying all of them from the right by $Z$ and finally tensoring with $X$ for qubit number $n+1$. As usual, we now write the string of unitaries associated to each row as $V^{(n+1)}_j$ for $n=1,\ldots,2^{n+1}$. We similarly use $V_j^{(n)}$ for the unitary strings for the case of $n$ qubits.

To see that this yields a basis, we first show that the first $2^{n}$ basis elements (upper block of table, $j=1,\ldots,2^{n}$) are orthogonal. For this purpose, we use the recursion formula \eqref{Wrecursion} to write for $j\neq j'$
\begin{align}\nonumber
	\bracket{W_{n+1}}{(V^{(n+1)}_{j'})^\dagger V^{(n+1)}_j}{W_{n+1}}&=\frac{n}{n+1} \bracket{W_n 0}{(V^{(n)}_{j'})^\dagger V_j^{(n)} \otimes \openone  }{W_n0}+\frac{1}{n+1} \bracket{0\ldots 01}{(V^{(n)}_{j'})^\dagger V_j^{(n)} \otimes \openone}{0\ldots 01}\\\nonumber
	&+\frac{\sqrt{n}}{n+1}\bracket{W_n 0}{(V^{(n)}_{j'})^\dagger V_j^{(n)} \otimes \openone}{0\ldots 01}+\frac{\sqrt{n}}{n+1}\bracket{0\ldots 01}{(V^{(n)}_{j'})^\dagger V_j^{(n)} \otimes \openone}{W_n 0}=0
\end{align}
The first term is zero for all $j'\neq j$ due to the induction hypothesis. The third and fourth terms are zero due to orthogonality in the last qubit register. The second term is zero for every $j'\neq j$ there exists at least one qubit register $k$ for which $U^{(j')}_k$ and $U^{(j)}_k$ are composed of different numbers of bit-flips ($X$). The latter follows from the initial condition of using $\{\id, X\}$ to construct the $\ket{W_1}$-basis. 

The same procedure will analogously show that the latter $2^n$ basis elements (lower block of the table, $j=2^{n}+1,\ldots,2^{n+1}$) are orthogonal. We are left with showing that every overlap between the upper and lower block, i.e.~with any $j'=1,\ldots,2^n$ and any $j=2^{n}+1,\ldots,2^{n+1}$, also vanishes. For this we have
\begin{align}\nonumber
	\bracket{W_{n+1}}{(V^{(n+1)}_{j'})^\dagger V^{(n+1)}_j}{W_{n+1}}&=\frac{n}{n+1} \bracket{W_n 0}{ \left[(V^{(n)}_{j'})^\dagger V_j^{(n)} \otimes X\right] \bigotimes_{k=1}^n Z\otimes \openone }{W_n0}\\\nonumber
	&+\frac{1}{n+1} \bracket{0\ldots 01}{ \left[(V^{(n)}_{j'})^\dagger V_j^{(n)} \otimes X\right] \bigotimes_{k=1}^n Z\otimes \openone}{0\ldots 01}\\\nonumber
	&+\frac{\sqrt{n}}{n+1}\bracket{W_n 0}{ \left[(V^{(n)}_{j'})^\dagger V_j^{(n)} \otimes X\right] \bigotimes_{k=1}^n Z\otimes \openone}{0\ldots 01}\\ \nonumber
	& +\frac{\sqrt{n}}{n+1}\bracket{0\ldots 01}{ \left[(V^{(n)}_{j'})^\dagger V_j^{(n)} \otimes X\right] \bigotimes_{k=1}^n Z\otimes \openone}{W_n 0}
\end{align}
Note that $\bigotimes_{k=1}^n Z\otimes \openone \ket{W_n0}=-\ket{W_n0}$ and $\bigotimes_{k=1}^n Z\otimes \openone \ket{0\ldots 01}=\ket{0\ldots 01}$. The first and second terms are both zero due to orthogonality in the final qubit register. We thus have
\begin{align}\nonumber
	\bracket{W_{n+1}}{(V^{(n+1)}_{j'})^\dagger V^{(n+1)}_j}{W_{n+1}}&=\frac{\sqrt{n}}{n+1}\bracket{W_n}{ (V^{(n)}_{j'})^\dagger V_j^{(n)}}{0\ldots 0}-\frac{\sqrt{n}}{n+1}\bracket{0\ldots 0}{(V^{(n)}_{j'})^\dagger V_j^{(n)}}{W_n}\\
	&= \frac{\sqrt{n}}{n+1}\bracket{W_n}{ (V^{(n)}_{j'})^\dagger V_j^{(n)}-(V^{(n)}_{j})^\dagger V_{j'}^{(n)}}{0\ldots 0}=0.
\end{align}
The last equality follows from the fact that it is sufficient, for given $(j,j')$, that there exist some register index $k$ such that $(U^{(j')})_k^\dagger U^{(j)}_k-(U^{(j)})_k^\dagger U^{(j')}_k=0$ in order for the overlap to vanish. This is always the case because due to our construction (see initial condition and the table), for every two unitaries there is at least one register $k$ where the single-qubit unitaries differ by $X$, meaning that either $(U^{(j)}_k,U^{(j')}_k)=(\openone,X) /(Z,XZ) $, or the same with $j \leftrightarrow j'$ is true. The condition above is satisfied by all of these combinations. Hence we conclude that the proposed construction satisfies
\begin{equation}
	\bracket{W_{n+1}}{(V_j^{(n+1)})^\dagger V_{j'}^{(n+1)}}{W_{n+1}}=\delta_{jj'}
\end{equation}
and  therefore yields a $W$-state basis for any number of qubits.

\section{The Pauli structure for state-independent qubit unitary constructions}\label{AppPauliStructure}

We consider the set of local unitaries $\mathcal{P}$ that are applied to the $i$-th qubit in the state-independent construction and show that without loss of generality, the set can be chosen to be the Pauli-type gates $\mathcal{P}\equiv \{\openone, X,Z,XZ\}$. First, note that the set is finite since there are exactly $2^n$ basis states. Next, we observe that the identity $\openone$ has to be within the set $\mathcal{P}$ since we demand that $V_1=\openone$. Furthermore, we can argue that the gate
\begin{equation}\nonumber
	XZ=\left(\begin{array}{rr}
		0 & -1 \\
		1 & 0 \\
	\end{array}\right)
\end{equation}
has to be within the set as well, since it is the only gate that maps every real qubit state to its orthogonal state. More precisely, if it is not used on the $i$-th qubit at least once, one can choose a real qubit state $\ket{\phi_{i}}$ such that none of the gates in $\mathcal{P}$ map $\ket{\phi_{i}}$ to its orthogonal vector. Hence if we apply the state-independent construction to the real-valued product state $\ket{\phi}=\ket{0}_1\otimes \ldots \ket{0}_{i-1}\otimes \ket{\phi_{i}}\otimes\ket{0}_{i+1}\otimes \ldots \otimes \ket{0}_{n}$ none of the resulting $2^n$ states are distinguishable on the $i$-th qubit, which is impossible if these states should form a basis of product states. Therefore, the gate $XZ$ has to be within the set $\mathcal{P}$. Apart from the gates $\openone$ and $XZ$ we can constrain which other qubit unitaries can be in the set $\mathcal{P}$. We know that if we demand $V_1=\openone$, every string of local unitaries $(V_j)$ and their products $(V_j)^\dagger V_{j'}$ with $j\neq j'$ have to be skew-symmetric. As a result, the local unitaries on each subsystem (hence, the unitaries in the set $\mathcal{P}$) and also all their products have to be either symmetric or skew-symmetric. By neglecting a global phase, the general form of a unitary operator can be written as:
\begin{equation}
	U=\left( \begin{array}{rr}
		\cos{\left(\theta\right)}e^{i\alpha} & \sin{\left(\theta\right)}e^{i\beta}  \\ 
		-\sin{\left(\theta\right)}e^{-i\beta} & \cos{\left(\theta\right)}e^{-i\alpha}  \\
	\end{array}\right).
\end{equation}
The only skew-symmetric $2\times2$ unitary is, up to an irrelevant global phase, the Pauli-type operator $XZ$, which we already found to be necessarily in the set $\mathcal{P}$. All the symmetric matrices of this form can be written as:
\begin{equation}
	U=\left( \begin{array}{rr}
		\cos{\left(\theta\right)}e^{i\alpha} & i\sin{\left(\theta\right)} \\ 
		i\sin{\left(\theta\right)} & \cos{\left(\theta\right)}e^{-i\alpha}  \\
	\end{array}\right).
\end{equation}
If the gate $U$ is in $\mathcal{P}$, it is at some point multiplied with the gate $XZ$ since the operator $XZ$ is used at least once on the $i$-th qubit. Since we know that the result of this product has to be again either symmetric or skew-symmetric, we obtain that $\alpha=\pi/2,3\pi/2$ due to:
\begin{align}
	(XZ)^\dagger U= \left(\begin{array}{rr}
		0 & 1 \\
		-1 & 0 \\
	\end{array}\right) \left( \begin{array}{rr}
		\cos{\left(\theta\right)}e^{i\alpha} & i\sin{\left(\theta\right)} \\ 
		i\sin{\left(\theta\right)} & \cos{\left(\theta\right)}e^{-i\alpha}  \\
	\end{array}\right)= \left( \begin{array}{rr}
		i\sin{\left(\theta\right)}   & \cos{\left(\theta\right)}e^{-i\alpha} \\ 
		-\cos{\left(\theta\right)}e^{i\alpha} & -i\sin{\left(\theta\right)}  \\
	\end{array}\right)\, .
\end{align}
The two possibilities for $\alpha=\pi/2,3\pi/2$ correspond to the two solutions
\begin{align}
	U_1=\left( \begin{array}{rr}
		\cos{\left(\theta\right)} & \sin{\left(\theta\right)} \\ 
		\sin{\left(\theta\right)} & -\cos{\left(\theta\right)}  \\
	\end{array}\right) \, , \quad U_2=\left( \begin{array}{rr}
		\sin{\left(\theta\right)}   & -\cos{\left(\theta\right)}\\ 
		-\cos{\left(\theta\right)}& -\sin{\left(\theta\right)}  \\
	\end{array}\right).
\end{align}
We left the irrelevant global factor $i$ for simplicity. Considering the additional degree of freedom of $\theta$, we can restrict to the first class of solutions $U_1$ since the second class $U_2$ can be obtained by shifting $\theta$ by $\pi /2$. Hence, if we add a gate $U$ to the set $\mathcal{P}$, it has to be of the form given by $U_1$ above. Now if we add two such gates to the set $\mathcal{P}$, the product of $U_1$ with another valid matrix $U'_1$ is
\begin{align*}
	U^\dagger_1 U'_1&=\left( \begin{array}{rr}
		\cos{\left(\theta\right)}& \sin{\left(\theta\right)} \\ 
		\sin{\left(\theta\right)} & -\cos{\left(\theta\right)}  \\
	\end{array}\right)
	\left( \begin{array}{rr}
		\cos{\left(\theta'\right)}& \sin{\left(\theta'\right)} \\ 
		\sin{\left(\theta'\right)} & -\cos{\left(\theta'\right)}  \\
	\end{array}\right)= \\
	&=\left( \begin{array}{rr}
		\cos{\left(\theta\right)} \cos{\left(\theta'\right)} + \sin{\left(\theta\right)}\sin{\left(\theta'\right)} \quad &
		\cos{\left(\theta\right)} \sin{\left(\theta'\right)}  - \sin{\left(\theta\right)} \cos{\left(\theta'\right)}\\ 
		\sin{\left(\theta\right)} \cos{\left(\theta'\right)} - \cos{\left(\theta\right)} \sin{\left(\theta'\right)}   \quad &
		\cos{\left(\theta\right)}\cos{\left(\theta'\right)}+ \sin{\left(\theta\right)}\sin{\left(\theta'\right)}  \\
	\end{array}\right)&=\left( \begin{array}{rr}
		\cos{\left(\theta-\theta'\right)} \quad &
		-\sin{\left(\theta-\theta'\right)}\\ 
		\sin{\left(\theta-\theta'\right)}   \quad &
		\cos{\left(\theta-\theta'\right)}  \\
	\end{array}\right)
\end{align*}
If both, $U_1$ and $U'_1$, are in $\mathcal{P}$, this product has to be again either symmetric, which is true if $\theta =\theta'$ or skew-symmetric, which is true if $\theta =\theta'+\pi/2$. (Note that, also $\theta =\theta'+\pi$ and $\theta =\theta'+3\pi/2$ are possible solutions but we do not have to consider them since they just differ by an irrelevant global factor of $(-1)$ in one of the two unitaries.) Hence, $U'_1$ is either $U_1$ or the unitary $U_2$ stated above. Hence, for each single-qubit subsystem, we can only use a set of operators $\mathcal{P}\equiv\{\openone, U_1, U_2, XZ\}$ for our basis construction.

In a final step, we can show that we can restrict also $\theta$. To see this, suppose a state-independent construction exists where we use the gates from the set $\mathcal{P}\equiv\{\openone, U_1, U_2, XZ\}$. Now consider the construction where each gate $U_1$ is replaced with $W^\dagger U_1 W$, each gate $U_2$ with $W^\dagger U_2 W$, each gate $XZ$ with $W^\dagger XZ W$ and each gate $\openone$ with $W^\dagger \openone W$, where: 
\begin{align}
	W=\left( \begin{array}{rr}
		\cos{\left(\alpha\right)}& -\sin{\left(\alpha\right)} \\ 
		\sin{\left(\alpha\right)} & \cos{\left(\alpha\right)}  \\
	\end{array}\right) \,
\end{align}
for some freely chosen parameter $\alpha$. This also has to be a state-independent construction for any state with real coefficients, since $W$ is a map from real states to real states, and all inner products between the basis states remain the same under this local transformation. Hence, if a state-independent construction exists with the gate set $\mathcal{P}\equiv\{\openone, U_1, U_2, XZ\}$, another state-independent construction with the gate set $\mathcal{P}'\equiv\{W^\dagger \openone W, W^\dagger U_1 W, W^\dagger U_2 W, W^\dagger XZ W\}$ has to exist as well. Choosing $\alpha=\theta /2$, the set $\mathcal{P}'\equiv\{W^\dagger \openone W, W^\dagger U_1 W,W^\dagger U_2 W,W^\dagger XZ W\}$ becomes exactly $\mathcal{P}'\equiv \{\openone, X,Z,XZ\}$, which concludes the proof.


\end{document}